\pdfoutput=1 

\documentclass{ws-rmp}

\usepackage{amsmath}
\usepackage{amssymb}
\usepackage{amsfonts}
\usepackage{bm}
\usepackage{mathtools}
\usepackage{graphicx}
\usepackage[numbers,sort&compress]{natbib}
\usepackage{physics}
\usepackage{comment}

\newcommand{\HA}{\mathcal{H}}

\begin{document}

\markboth{G. Kimura}{Hiai--Petz Skew Informations and Sharp Uncertainty Relations}

\title{Hiai--Petz Skew Informations and Sharp Uncertainty Relations for a Single Observable}

\author{GEN KIMURA}
\address{
Graduate School of Information Sciences, Tohoku University,\\
6-3-09 Aoba, Aramaki-aza, Aoba-ku, Sendai 980-8579, Japan\\
gen.kimura.quant@gmail.com
}

\maketitle


\begin{abstract}
We study intrinsic quantum uncertainty associated with a single observable from the viewpoint of skew information. By relaxing ordinary additivity while retaining the other fundamental uncertainty requirements, we introduce a class of Hiai--Petz nonadditive skew informations that contains Hansen's metric-adjusted skew informations as the boundary case $\theta=1$.
Our main quantitative result is a sharp single-observable uncertainty relation for the full Hiai--Petz class. For a fixed state, the optimal coefficient depends only on its largest and smallest eigenvalues and remains sharp even after the maximal classical contribution to the variance is retained. The general theorem yields, as special cases, sharp bounds for metric-adjusted skew informations, the SLD quantum Fisher information, and the power-commutator family.
We further show that $K_s(\rho,A)=\frac12\|[\rho^s,A]\|_{\mathrm{HS}}^2$, $1/2\le s<1$, is realized within the Hiai--Petz framework through Stolarsky operator means, thereby establishing its convexity for $1/2<s<1$.
More generally, for each fixed operator monotone function $f$, the Hiai--Petz spectral kernel interpolates geometrically between the metric-adjusted skew information $I^f$ and the common endpoint $K_1$. The resulting family obeys an exact power-trace-weighted composition law, replacing ordinary additivity in the interior $0<\theta<1$.
\end{abstract}



\section{Introduction}

It is a great pleasure and honor to contribute this article to the special volume celebrating the eightieth birthday of Professor Fumio Hiai.
Professor Hiai has made profound and lasting contributions to operator theory, operator algebras, matrix analysis, quantum probability, and quantum information theory.
I have always been particularly impressed by his ability to identify the essential structure of a difficult mathematical problem and to reveal a direct and elegant path to its resolution.
This clarity has shaped not only his research contributions but also his remarkable educational influence, helping generations of students and researchers enter sophisticated areas of operator theory and matrix analysis through their essential ideas.

Although I was not formally a student of Professor Hiai, I learned much of the foundation of operator theory from his books, most notably his Japanese textbook \textit{Hilbert Spaces and Linear Operators} \cite{HiaiYanagi1995}.
For me, as for many students and researchers in Japan, this book provided an accessible yet rigorous introduction to Hilbert space theory, spectral theory, and linear operators.
His later books and monographs, including \textit{Introduction to Matrix Analysis and Applications} with D\'enes Petz \cite{HiaiPetz2014}, have continued to influence my understanding of matrix analysis, operator inequalities, and their connections with quantum information theory.

I was fortunate to belong to the same mathematics group as Professor Hiai during my time as a postdoctoral researcher at Tohoku University from 2005 to 2008, and I have benefited from his advice and kind support on many occasions since then.
It is therefore a particular honor for me to have returned to Tohoku University in 2026 and once again to belong to the same academic community.
Although we have not had the opportunity to collaborate directly, my own research has repeatedly crossed paths with mathematical themes developed by Professor Hiai and with researchers closely connected to him.
I have also had the privilege of seeing firsthand his generosity, modesty, and openness toward younger researchers.

The subject of the present article lies at the intersection of several themes central to Professor Hiai's work: operator monotone functions, operator means, quantum information geometry, convexity, and noncommutativity.
In this sense, the present study is also an opportunity to revisit these ideas from the viewpoint of quantum uncertainty and skew information \cite{WignerYanase1963}.

Quantum uncertainty is most commonly formulated as a trade-off involving two or more observables.
Skew information has played an important role in this direction.
Beginning with Luo's Wigner--Yanase-based uncertainty relation \cite{Luo2003}, this line of research has since been developed in many directions; see, e.g., \cite{LuoZhang2004,Kosaki2005,YanagiFuruichiKuriyama2005,Yanagi2010,Furuichi2010,GibiliscoIsola2007,Yanagi2011,FuruichiYanagi2012,YanagiFuruichiKuriyama2013,ChenFeiLong2016,RenLiYeLi2021}.
These results primarily concern trade-offs among the quantum uncertainties associated with different observables.

There is, however, a complementary question that already arises for a single observable: how much of its variance in a given state should be regarded as intrinsically quantum?
For a mixed state, the variance contains uncertainty originating from classical statistical mixing as well as uncertainty generated by the noncommutativity between the state and the observable.
The Wigner--Yanase skew information was introduced from an information-content viewpoint sensitive to this noncommutativity \cite{WignerYanase1963}, and Luo later emphasized its role as a quantum contribution to uncertainty in a mixed state \cite{LuoQuantumClassical,Luo2005,Luo2006}.
From this viewpoint, an intrinsic quantum uncertainty $Q_\rho(A)$ that agrees with the variance on pure states and is convex in the state automatically satisfies the single-observable bound
\begin{equation}
V_\rho(A)\ge Q_\rho(A).
\label{eq:intro-single-observable-UR}
\end{equation}
In our recent work \cite{YamashitaMayumiKimura2026}, we studied inequalities of this form as \emph{single-observable uncertainty relations} and showed that, for a fixed mixed state, the coefficient of $Q_\rho(A)$ can often be sharpened beyond the universal value one.
We also showed that the same optimal coefficient can be retained after adding a maximal classical contribution to the variance.
The present article is primarily concerned with this single-observable aspect of skew information.

The Wigner--Yanase--Dyson (WYD) family provides a fundamental extension of the original skew information, with convexity established by Lieb in his proof of the Wigner--Yanase--Dyson conjecture \cite{Lieb1973}.
A much broader structure emerges from quantum information geometry.
Petz's classification of monotone quantum Fisher metrics in terms of standard operator monotone functions \cite{Petz1996,Petz2002} leads to Hansen's metric-adjusted skew information \cite{Hansen2008}, which contains the WY and WYD skew informations as well as one quarter of the symmetric logarithmic derivative (SLD) quantum Fisher information.
Metric-adjusted skew informations have also found applications in quantum correlations and quantum metrology \cite{GibiliscoGirolamiHansen2021}, while quasi-entropies and their extensions developed by Hiai and Petz provide a flexible framework in which convexity, operator means, and noncommutative information quantities can be treated beyond the metric setting \cite{HiaiPetz2012,HiaiPetz2013}.

Within this single-observable framework, ordinary tensor-product additivity, although natural in the original Wigner--Yanase formulation, is not indispensable.
We retain instead five properties that are directly relevant to intrinsic quantum uncertainty: positivity, faithfulness to the noncommutativity $[\rho,A]$, agreement with the variance on pure states, convexity under mixing, and unitary covariance.
This gives a natural setting in which to ask how far skew-information-type uncertainty can be extended beyond Hansen's metric-adjusted class.

Our first result is that a normalized commutator restriction of the Hiai--Petz quasi-entropy-type functional gives such an extension.
For every regular standard operator monotone function $f$ and $0<\theta\le1$, we define a Hiai--Petz nonadditive skew information $\mathcal I_\rho^{f,\theta}(A)$.
The Hiai--Petz joint convexity theorem implies that it satisfies all five uncertainty requirements, while Hansen's metric-adjusted skew information is recovered exactly at the boundary $\theta=1$.

A central example is provided by the power-commutator family
\begin{equation}
K_s(\rho,A)=\frac12\|[\rho^s,A]\|_{\mathrm{HS}}^2,
\quad
\frac12\le s\le1.
\end{equation}
It continuously connects the Wigner--Yanase skew information $K_{1/2}=I^{\mathrm{WY}}$ to the Hilbert--Schmidt commutator $K_1$.
We show that, for $1/2\le s<1$, $K_s$ is itself a Hiai--Petz skew information, realized through a family of Stolarsky operator means.
In particular, this representation establishes the nontrivial convexity of $K_s$ in the intermediate range $1/2<s<1$.
For $1/2<s\le1$, the family lies outside Hansen's metric-adjusted class as a dimension-independent construction, as is also reflected in its nonadditive tensor-product composition law.

The Hiai--Petz realization of $K_s$ leads to a broader interpolation picture.
If $k_f$ denotes the pairwise spectral kernel of Hansen's metric-adjusted skew information and $k_{K_1}(x,y)=(x-y)^2$ is that of $K_1$, then
\begin{equation}
k_{f,\theta}(x,y)
=
k_{K_1}(x,y)^{1-\theta}
k_f(x,y)^\theta.
\label{eq:intro-kernel-interpolation}
\end{equation}
Thus, after adjoining the continuous endpoint $\theta=0$, every fixed metric-adjusted skew information is connected to the same quantity $K_1$ by a pairwise geometric interpolation.
The power family has a distinguished position in this structure: for each fixed $s$, $K_s$ lies on the Hiai--Petz interpolation path associated with a Stolarsky representing function $f_s$, while, as $s$ varies, both $f_s$ and the interpolation parameter $\theta_s=2(1-s)$ change.
Hence the power family traces a curve through the Hiai--Petz class from the Wigner--Yanase skew information to the common endpoint $K_1$, rather than following a single fixed-$f$ path.

A particularly important fixed-$f$ path is obtained from the SLD representing function.
In this case the Hiai--Petz parameter continuously connects $K_1$ to one quarter of the SLD quantum Fisher information.
This example is especially significant because one quarter of the SLD quantum Fisher information is the maximal intrinsic quantum uncertainty among functionals satisfying pure-state normalization and convexity.
The corresponding interpolation therefore connects the common endpoint $K_1$ directly to the strongest universal single-observable quantum uncertainty within the present framework.

The nonadditive character of the extension is controlled rather than arbitrary.
We derive an exact power-trace-weighted composition law for product states and additive observables.
At $\theta=1$ it reduces to the ordinary additivity of metric-adjusted skew information, whereas for $0<\theta<1$ it yields a nontrivial scaling under the addition of a mixed ancillary state.
This shows, in particular, that the interior Hiai--Petz family lies outside Hansen's metric-adjusted class as a dimension-independent family.

The second and principal part of the article concerns the sharp form of the single-observable uncertainty relation introduced above.
For a fixed mixed state, we determine the largest coefficient multiplying $\mathcal I_\rho^{f,\theta}(A)$ in a variance bound.
Our main result solves this optimization problem for the full Hiai--Petz class.
For every regular standard operator monotone function $f$ and every $0<\theta\le1$, the optimal state-dependent coefficient is obtained explicitly and, remarkably, depends only on $\lambda_{\min}$ and $\lambda_{\max}$, although $\mathcal I_\rho^{f,\theta}(A)$ itself generally depends on the full spectrum of $\rho$.
The same coefficient remains sharp after the maximal classical contribution to the variance is retained.
Specialization of the general theorem yields sharp uncertainty relations for Hansen's metric-adjusted skew informations and for the nonadditive Hiai--Petz family, including the fixed-SLD interpolation and the power-commutator family.

The paper is organized as follows.
In Sec.~\ref{sec:single-observable-uncertainty}, we formulate intrinsic quantum uncertainty and explain why additivity is not imposed as a defining axiom.
We review Hansen's metric-adjusted skew information and introduce the Hiai--Petz skew information.
We then realize the power-commutator family through Stolarsky operator means and establish its position within the general Hiai--Petz interpolation structure.
We next examine the fixed-SLD interpolation to one quarter of the SLD quantum Fisher information and derive the power-trace-weighted composition law that distinguishes the interior Hiai--Petz family from the metric-adjusted class.
In Sec.~\ref{sec:optimal-single-observable}, we recall the maximal classical contribution to the variance and prove the sharp fixed-state Hiai--Petz uncertainty relation.
We then specialize the result to Hansen's class, the fixed-SLD interpolation, and the power-commutator family.
Section~\ref{sec:conclusion} summarizes the results and discusses the spectral structure underlying the optimal bounds.

\section{Intrinsic quantum uncertainty and Hiai--Petz skew information}
\label{sec:single-observable-uncertainty}

\subsection{Axioms on intrinsic quantum uncertainty}

The starting point of the present work is to regard skew information (SI) as a measure of the intrinsic quantum uncertainty associated with a single observable. 
The idea goes back to the original information-content interpretation of Wigner and Yanase \cite{WignerYanase1963}. 
It is also closely related to Luo's observation that, for a mixed state, the variance contains two conceptually different contributions: classical uncertainty originating from statistical mixing and intrinsic quantum uncertainty \cite{LuoQuantumClassical,Luo2005,Luo2006}. 
Motivated by these viewpoints, we introduce the following notion of intrinsic quantum uncertainty, which essentially corresponds to Luo's notion of \emph{quantum uncertainty} \cite{LuoQuantumClassical}. 
\begin{definition}
A functional $Q_\rho(A)$ is called an \emph{intrinsic quantum uncertainty} for a state $\rho$ and an observable $A$ if it satisfies the following properties (Q1)-(Q5):  
\begin{enumerate}
\item[(Q1)] \textbf{Positivity:}
\begin{equation}
Q_\rho(A)\ge0.
\label{eq:Q1}
\end{equation}

\item[(Q2)] \textbf{Faithfulness to noncommutativity:}
\begin{equation}
Q_\rho(A)=0
\quad\Longleftrightarrow\quad
[\rho,A]=0.
\label{eq:Q2}
\end{equation}

\item[(Q3)] \textbf{Pure-state normalization:}
for every pure state $P = \ketbra{\phi}{\phi}$, 
\begin{equation}
Q_P(A)=V_P(A).
\label{eq:Q3}
\end{equation}

\item[(Q4)] \textbf{Convexity in the state:}
for $0\le p\le1$ and $\rho_1,\rho_2 \in {\cal D}(\HA)$, 
\begin{equation}
Q_{p\rho_1+(1-p)\rho_2}(A)
\le
pQ_{\rho_1}(A)+(1-p)Q_{\rho_2}(A).
\label{eq:Q4}
\end{equation}

\item[(Q5)] \textbf{Unitary covariance:}
for every unitary operator $U$,
\begin{equation}
Q_{U\rho U^\dagger}(UAU^\dagger)
=
Q_\rho(A).
\label{eq:Q5}
\end{equation}
\end{enumerate}
\end{definition}
Each of the above requirements has a direct physical interpretation.
Condition (Q1) is the minimal requirement for an uncertainty measure. 
Condition (Q2) expresses that $Q_\rho(A)$ quantifies uncertainty specifically associated with the noncommutativity between the state and the observable. 
Indeed, it is elementary to show that if $\rho$ and $A$ do not commute, then $V_\rho(A) > 0$ (See e.g., \cite{YamashitaMayumiKimura2026}). 
Condition (Q3) expresses an important consistency requirement: For a pure state there is no uncertainty arising from classical mixing, and hence the intrinsic quantum uncertainty should coincide with the full variance. 
This is also central to Luo's interpretation of skew information as the genuinely quantum part of uncertainty
\cite{Luo2005,Luo2006}. 
Condition (Q4) is particularly important, since it requires that probabilistic mixing of states cannot increase the intrinsic quantum part of the uncertainty. 
This is the counterpart, in the present setting, of the original Wigner--Yanase requirement that information should decrease under mixing.
Condition (Q5) ensures that the quantity is independent of the particular Hilbert-space basis used to represent the state and observable. 
In particular, it contains the invariance under isolated dynamics considered by Wigner and Yanase whenever $A$ is a conserved observable.

\begin{remark}\label{rem:Tsallis-analogy} 
In the original work of Wigner and Yanase, additivity was also imposed as a natural requirement on the information content:
\begin{enumerate}
\item[(Q6)] \textbf{Additivity:} For product states and additive observables,
\begin{equation}
Q_{\rho_1\otimes\rho_2}(A_1\otimes I+I\otimes A_2)=Q_{\rho_1}(A_1)+Q_{\rho_2}(A_2).
\label{eq:Q6-additivity}
\end{equation}
\end{enumerate}
For the purpose of single-observable uncertainty, however, condition (Q6) is not essential. We therefore consider a broader class in which ordinary additivity is not imposed as a defining axiom, while conditions (Q1)--(Q5) are retained. For this reason, we shall also refer to such an intrinsic quantum uncertainty as a \emph{nonadditive extension of skew information}.

There is a useful analogy with nonadditive generalizations of entropy. The (quantum) Tsallis entropy \cite{Tsallis1988},
\begin{equation}
S_q(\rho)=\frac{1-\Tr\rho^q}{q-1},\quad q>0,\quad q\ne1,
\label{eq:Tsallis}
\end{equation}
reduces to the von Neumann entropy in the limit $q\to1$, while ordinary additivity is replaced by the \emph{pseudo-additivity} law
\begin{equation}
S_q(\rho\otimes\sigma)=S_q(\rho)+S_q(\sigma)+(1-q)S_q(\rho)S_q(\sigma).
\label{eq:Tsallis-pseudoadditivity}
\end{equation}
Thus, the loss of ordinary additivity does not imply the loss of a controlled composition law. As we shall see later, a closely analogous phenomenon occurs for the Hiai--Petz family studied below, where ordinary additivity is replaced by a deformed composition law; see Sec.~\ref{sec:deformed-additivity}.
\end{remark}

\begin{remark}
Wigner and Yanase considered an even stronger composition requirement.
For a bipartite state $\rho_{12}$, with
$\rho_j=\Tr_{\bar j}\rho_{12}$ and
\begin{equation}
A_{12}
=
A_1\otimes I+I\otimes A_2,
\end{equation}
the proposed superadditivity property is
\begin{equation}
Q_{\rho_{12}}(A_{12})
\ge
Q_{\rho_1}(A_1)+Q_{\rho_2}(A_2).
\label{eq:WY-superadditivity}
\end{equation}
Wigner and Yanase proved this relation when the joint state is pure, but the conjecture fails for general mixed states \cite{Hansen2007}. 
\end{remark}

\subsubsection{Examples} The prototype of intrinsic quantum uncertainty is, of course, the Wigner--Yanase (WY) skew information \cite{WignerYanase1963}. Using the eigenvalue notation introduced above, the WY skew information can be written as
\begin{equation}
I_\rho^{\mathrm{WY}}(A):=-\frac12\Tr[\sqrt{\rho},A]^2
=
\sum_{i<j}
(\sqrt{\lambda_i}-\sqrt{\lambda_j})^2
|A_{ij}|^2.
\label{WY}
\end{equation}
Its one-parameter extension, known as the Wigner--Yanase--Dyson (WYD) skew information, provides a further family of examples:
\begin{equation}
I_{\rho,\alpha}^{\mathrm{WYD}}(A):=-\frac12\Tr[\rho^\alpha,A][\rho^{1-\alpha},A]
=
\sum_{i<j}
(\lambda_i^\alpha-\lambda_j^\alpha)
(\lambda_i^{1-\alpha}-\lambda_j^{1-\alpha})
|A_{ij}|^2,\quad 0<\alpha<1.
\label{eq:WYD}
\end{equation}
A fundamental result of Lieb \cite{Lieb1973}, originally proving the Wigner--Yanase--Dyson conjecture, established the convexity of $I_{\rho,\alpha}^{\mathrm{WYD}}(A)$ with respect to the state $\rho$ for every $0<\alpha<1$.

Another important example is the symmetric logarithmic derivative (SLD) quantum Fisher information \cite{BraunsteinCaves1994}:
\begin{equation}
\frac14F_Q^{\mathrm{SLD}}(\rho,A)
=
\sum_{\substack{i<j\\ \lambda_i+\lambda_j>0}}
\frac{(\lambda_i-\lambda_j)^2}{\lambda_i+\lambda_j}|A_{ij}|^2.
\label{eq:SLD-QFI-quarter}
\end{equation}
The factor $1/4$ is chosen so that the pure-state normalization (Q3) is satisfied. 
Finally, these examples are encompassed by Hansen's framework of metric-adjusted skew information \cite{Hansen2008}, which will be reviewed in detail later. 

A particularly simple example that lies outside Hansen's metric-adjusted class in general dimensions\footnote{For a single qubit, $K_1$ coincides with one quarter of the SLD quantum Fisher information because $\lambda_1+\lambda_2=1$. This coincidence, however, is specific to dimension two and is not stable under tensor products.} is
\begin{equation}
K_1(\rho,A):=\frac12\|[\rho,A]\|_{\mathrm{HS}}^2
=
\sum_{i<j}(\lambda_i-\lambda_j)^2|A_{ij}|^2.
\label{eq:K1-intro}
\end{equation}
This Hilbert--Schmidt commutator quantity has also appeared previously in the literature.
In particular, Girolami \cite{Girolami2014} introduced
$-\frac14\Tr[\rho,A]^2=K_1(\rho,A)/2$
as an experimentally accessible lower bound on the Wigner--Yanase skew information. 
In the present work, we regard the normalization \eqref{eq:K1-intro} as a skew information in its own right, namely, as an intrinsic quantum uncertainty satisfying (Q1)--(Q5). 
In particular, its convexity in $\rho$ follows immediately from the linearity of $\rho\mapsto[\rho,A]$ and the convexity of the squared Hilbert--Schmidt norm. 
As we show below, $K_1$ plays a distinguished structural role in the present framework: it emerges as the common continuous endpoint from which the Hiai--Petz family interpolates to Hansen's metric-adjusted skew informations; see Sec.~\ref{sec:interpolation-structure}.

A natural generalization is obtained by replacing $\rho$ with a fractional power $\rho^s$. For $1/2\le s\le1$, we define
\begin{equation}
K_s(\rho,A):=\frac12\|[\rho^s,A]\|_{\mathrm{HS}}^2
=
\sum_{i<j}(\lambda_i^s-\lambda_j^s)^2|A_{ij}|^2.
\label{eq:Ks-intro}
\end{equation}
The endpoint $s=1$ recovers \eqref{eq:K1-intro}, whereas $s=1/2$ coincides with the Wigner--Yanase skew information. The requirements (Q1)--(Q3) and (Q5) are immediate for $K_s$. The remaining convexity condition (Q4) for $1/2<s<1$ requires a separate argument, since the standard methods based on Lieb's concavity theorem \cite{Lieb1973}, Effros' perspective method \cite{Effros2009}, and Hansen's theory of metric-adjusted skew information \cite{Hansen2008} do not directly apply.
The key observation of the present work is that the intermediate regime $1/2<s<1$ can instead be placed within the Hiai--Petz theory of quasi-entropy-type functionals \cite{HiaiPetz2013}. This will identify \eqref{eq:Ks-intro} as a concrete example of the normalized Hiai--Petz skew information introduced below.
For $1/2<s\le1$, these quantities lie outside Hansen's metric-adjusted class as a dimension-independent family. Indeed, every metric-adjusted skew information satisfies the ordinary additivity property (Q6), whereas, as shown in Sec.~\ref{sec:deformed-additivity}, $K_s$ instead obeys a different, power-trace-weighted composition law.

\subsubsection{Single-observable uncertainty relations}
An important consequence of (Q3) and (Q4) is that every intrinsic quantum uncertainty $Q_\rho(A)$ provides a universal lower bound on the variance.
Indeed, the variance is concave in the state, since the first term in
\begin{equation}
V_\rho(A)=\Tr(\rho A^2)-\bigl(\Tr(\rho A)\bigr)^2
\end{equation}
is linear in $\rho$, while the second term is concave as the composition of the linear functional $\rho\mapsto\Tr(\rho A)$ with the concave function $x\mapsto -x^2$.
Hence, using $\rho=\sum_i\lambda_i\ketbra{i}{i}$, we obtain
\begin{equation}
V_\rho(A)
\ge
\sum_i\lambda_i V_{\ketbra{i}{i}}(A)
\overset{\mathrm{(Q3)}}{=}
\sum_i\lambda_i Q_{\ketbra{i}{i}}(A)
\overset{\mathrm{(Q4)}}{\ge}
Q_\rho(A).
\label{eq:Q-variance-bound}
\end{equation}
Thus every intrinsic quantum uncertainty automatically gives rise to an uncertainty relation of the form
\begin{equation}
V_\rho(A)\ge Q_\rho(A).
\label{eq:single-observable-UR}
\end{equation}
In our recent work \cite{YamashitaMayumiKimura2026}, inequalities of this type were interpreted as \emph{single-observable uncertainty relations}, since the lower bound is determined entirely by the state $\rho$ and the single observable $A$, without introducing a second observable.

The maximal single-observable quantum uncertainty measure within the present framework is one quarter of the SLD quantum Fisher information. T\'oth and Petz conjectured, and established in several special cases, that this quantity is given by the convex roof of the variance \cite{TothPetz2013}; the conjecture was subsequently proved in full generality by Yu \cite{Yu2013}:
\begin{equation}
\frac14F_Q^{\mathrm{SLD}}(\rho,A)
=
\inf_{\rho=\sum_k p_k P_k}
\sum_k p_k V_{P_k}(A),
\label{eq:Yu}
\end{equation}
where the infimum is taken over all pure-state decompositions of $\rho$. Consequently, any functional satisfying (Q3) and (Q4) obeys
\begin{equation}
Q_\rho(A)
\le
\frac14F_Q^{\mathrm{SLD}}(\rho,A)
\le
V_\rho(A).
\label{eq:Q-QFI-bound}
\end{equation}
In this sense, one quarter of the SLD quantum Fisher information is the maximal single-observable quantum uncertainty measure, or equivalently, it provides the strongest universal lower bound on the variance among all functionals satisfying (Q3) and (Q4).

\subsection{Metric-adjusted skew information and Hiai--Petz skew information}

\subsubsection{Hansen's metric-adjusted skew information}

Before introducing the Hiai--Petz skew information, we first recall Hansen's metric-adjusted skew information \cite{Hansen2008}. 
Recall that a function $f:(0,\infty)\to(0,\infty)$ is called operator monotone if $0<X\le Y$ implies $f(X)\le f(Y)$ for positive definite operators $X$ and $Y$; see, e.g., \cite{Loewner1934,Donoghue1974,Bhatia2007}. 
A standard operator monotone function is an operator monotone function satisfying
\begin{equation}
f(1)=1,\quad f(t)=t f(t^{-1}).
\label{eq:standard-f}
\end{equation}
We say that $f$ is regular if
\begin{equation}
f(0):=\lim_{t\downarrow0}f(t)>0.
\label{eq:regular-f}
\end{equation}
The associated scalar mean is defined by
\begin{equation}\label{eq:Mf}
M_f(x,y):=y f(x/y),\quad x,y>0.
\end{equation}
By the symmetry condition in \eqref{eq:standard-f}, $M_f$ is symmetric:
\begin{equation}
M_f(x,y)=M_f(y,x).
\end{equation}
For regular $f$, $M_f$ admits a continuous extension to $[0,\infty)^2$, given by
\begin{equation}
        M_f(x,0)=M_f(0,x)=xf(0),\quad x\ge0.
\label{eq:Mf-boundary}
\end{equation}
For a faithful state $\rho$, let $L_\rho$ and $R_\rho$ denote the left- and right-multiplication superoperators,
\begin{equation}
L_\rho(X):=\rho X,\quad R_\rho(X):=X\rho.
\label{eq:left-right}
\end{equation}
By the Kubo--Ando correspondence \cite{KuboAndo1980}, every standard operator monotone function $f$ determines an operator mean, whose scalar restriction is $M_f$. 
Since $L_\rho$ and $R_\rho$ commute, the corresponding operator mean can be written as
\begin{equation}
M_f(L_\rho,R_\rho)=R_\rho f(L_\rho R_\rho^{-1}).
\label{eq:Mf-superoperator}
\end{equation}
The monotone quantum Fisher metric associated with $f$ is then given by
\begin{equation}
\gamma_\rho^f(X,Y):=\left\langle X,M_f(L_\rho,R_\rho)^{-1}(Y)\right\rangle_{\mathrm{HS}}
=\Tr X^\dagger M_f(L_\rho,R_\rho)^{-1}(Y).
\label{eq:metric}
\end{equation}
Petz's classification theorem identifies such metrics with standard operator monotone functions \cite{Petz1996,Petz2002}.

For a regular $f$, Hansen's metric-adjusted skew information \cite{Hansen2008} is defined by
\begin{equation}
I_\rho^f(A):=\frac{f(0)}{2}\gamma_\rho^f\bigl(i[\rho,A],i[\rho,A]\bigr).
\label{eq:MASI-intro}
\end{equation}
The commutator direction appearing here has a natural geometric interpretation. 
Indeed, the observable $A$ generates the unitary orbit
$\rho_t=e^{-itA}\rho e^{itA}$ through $\rho$, whose tangent vector at $t=0$ is
\begin{equation}
\left.\frac{d}{dt}\rho_t\right|_{t=0}=i[\rho,A].
\end{equation}
Thus, up to the normalization factor $f(0)/2$, the metric-adjusted skew information is the squared length, with respect to the monotone metric $\gamma_\rho^f$, of the tangent vector generated by $A$. In particular, it vanishes precisely when the unitary orbit is stationary at $\rho$, that is, when $[\rho,A]=0$.

Using the notation introduced above, this can equivalently be written as
\begin{equation}
I_\rho^f(A)=f(0)\sum_{\substack{i<j\\ \lambda_i+\lambda_j>0}}
\frac{(\lambda_i-\lambda_j)^2}{M_f(\lambda_i,\lambda_j)}|A_{ij}|^2.
\label{eq:MASI-spectral}
\end{equation}
For non-faithful states, \eqref{eq:MASI-spectral} is understood as the continuous extension from faithful states; equivalently, pairs with $\lambda_i=\lambda_j=0$ do not contribute.
Notice that Hansen's results \cite{Hansen2008} imply that metric-adjusted skew information satisfies all the requirements (Q1)--(Q5), as well as the additivity property (Q6).

The Wigner--Yanase and Wigner--Yanase--Dyson skew informations, as well as one quarter of the SLD quantum Fisher information, all arise as special cases of Hansen's metric-adjusted skew information. 
The WY skew information corresponds to 
\begin{equation}
f_{\mathrm{WY}}(t):=\frac{(1+\sqrt{t})^2}{4},
\quad
f_{\mathrm{WY}}(0)=\frac14.
\label{eq:f-WY}
\end{equation}
More generally, the WYD skew information with $0<\alpha<1$ corresponds to
\begin{equation}
f_{\alpha}^{\mathrm{WYD}}(t):=
\alpha(1-\alpha)
\frac{(t-1)^2}{(t^\alpha-1)(t^{1-\alpha}-1)},
\quad
f_{\alpha}^{\mathrm{WYD}}(0)=\alpha(1-\alpha).
\label{eq:f-WYD}
\end{equation}
The WY case is recovered at $\alpha=1/2$. Finally, one quarter of the SLD quantum Fisher information corresponds to the arithmetic-mean function
\begin{equation}
f_{\mathrm{SLD}}(t):=\frac{1+t}{2},
\quad
f_{\mathrm{SLD}}(0)=\frac12.
\label{eq:f-SLD}
\end{equation}
Substituting these functions into the scalar mean $M_f(x,y)=y f(x/y)$ in \eqref{eq:Mf}, the spectral representation \eqref{eq:MASI-spectral} immediately recovers the original expressions \eqref{WY}, \eqref{eq:WYD}, and \eqref{eq:SLD-QFI-quarter}.

\subsubsection{Hiai-Petz skew information}

We now turn from Hansen's metric-adjusted framework to the broader class of quasi-entropy-type functionals introduced by Hiai and Petz \cite{HiaiPetz2013}. For positive definite operators $P,Q$, it is convenient to introduce the positive sesquilinear form
\begin{equation}
\gamma_{P,Q}^{f,\theta}(X,Y)
:=
\Tr X^\dagger M_f(L_P,R_Q)^{-\theta}(Y),
\quad 0<\theta\le1.
\label{eq:HP-sesquilinear}
\end{equation}
The Hiai--Petz functional is the associated quadratic form
\begin{equation}
I_f^\theta(P,Q,X)
=
\gamma_{P,Q}^{f,\theta}(X,X).
\label{eq:HP-original}
\end{equation}
In particular, when $\theta=1$ and $P=Q=\rho$, the above sesquilinear form reduces to Petz's monotone metric:
\begin{equation}
\gamma_{\rho,\rho}^{f,1}(X,Y)
=
\gamma_\rho^f(X,Y).
\label{eq:HP-metric-reduction}
\end{equation}
The joint convexity theorem of Hiai and Petz implies that, for operator monotone $f$ and $0<\theta\le1$,
\begin{equation}\label{eq:HPJconvex}
(P,Q,X)\longmapsto I_f^\theta(P,Q,X)
\end{equation}
is jointly convex. Thus, the case $\theta=1$ recovers the monotone-metric structure when $P=Q$, whereas $0<\theta<1$ provides a genuine extension beyond it.

Motivated by Hansen's construction, we normalize the commutator restriction of \eqref{eq:HP-original} as
\begin{equation}
\mathcal I_\rho^{f,\theta}(A):=\frac{f(0)^\theta}{2}I_f^\theta\bigl(\rho,\rho,i[\rho,A]\bigr),\quad 0<\theta\le1.
\label{eq:HPskew-intro}
\end{equation}
In this paper, we shall refer to this normalized commutator restriction as the \emph{Hiai--Petz nonadditive skew information}, or, for short, the \emph{HP skew information}. Here, the term ``nonadditive'' indicates that the additivity property (Q6) is not imposed and is generally lost within this broader class. Instead, as we shall show in Sec.~\ref{sec:deformed-additivity}, the entire Hiai--Petz class satisfies a deformed additivity law.

Using the same eigenvalue notation, HP skew information can equivalently be written as
\begin{equation}
\mathcal I_\rho^{f,\theta}(A)=f(0)^\theta\sum_{\substack{i<j\\ \lambda_i+\lambda_j>0}}\frac{(\lambda_i-\lambda_j)^2}{M_f(\lambda_i,\lambda_j)^\theta}|A_{ij}|^2.
\label{eq:HPskew-spectral}
\end{equation}
For non-faithful states, \eqref{eq:HPskew-spectral} is understood as the continuous extension from faithful states; equivalently, pairs with $\lambda_i=\lambda_j=0$ do not contribute.

At $\theta=1$, \eqref{eq:HPskew-intro} reduces exactly to Hansen's metric-adjusted skew information \eqref{eq:MASI-intro}, and correspondingly \eqref{eq:HPskew-spectral} reduces to \eqref{eq:MASI-spectral}. For $0<\theta<1$, however, the Hiai--Petz construction goes beyond the metric-adjusted framework and yields a broad family of intrinsic quantum uncertainty measures. We first show that the power-commutator quantities \eqref{eq:Ks-intro} form a concrete family within this framework. This realization will then lead naturally to a general interpolation interpretation of the parameter $\theta$.

\begin{proposition}
\label{prop:HP-Q1-Q5}
For every regular standard operator monotone function $f$ and every $0<\theta\le1$, the normalized Hiai--Petz skew information \eqref{eq:HPskew-intro} satisfies the requirements (Q1)--(Q5).
\end{proposition}

\begin{proof}
Since $f$ is regular, Eq.~\eqref{eq:regular-f} implies that the associated operator mean admits the boundary value \eqref{eq:Mf-boundary}. 
By symmetry of the mean, we also have $M_f(0,x)=xf(0)$. Hence $M_f(x,y)>0$ whenever $x+y>0$. In particular, every denominator appearing in \eqref{eq:HPskew-spectral} is strictly positive for those pairs $(i,j)$ with $(\lambda_i-\lambda_j)^2\ne0$.

Property (Q1) now follows directly from \eqref{eq:HPskew-spectral}. Indeed, $f(0)^\theta>0$, $(\lambda_i-\lambda_j)^2\ge0$, $M_f(\lambda_i,\lambda_j)^\theta>0$ whenever the numerator is nonzero, and $|A_{ij}|^2\ge0$. 

For (Q2), since $M_f(\lambda_i,\lambda_j)>0$ whenever $\lambda_i+\lambda_j>0$, Eq.~\eqref{eq:HPskew-spectral} shows that
\begin{equation}
\mathcal I_\rho^{f,\theta}(A)=0
\end{equation}
if and only if $A_{ij}=0$ for every pair satisfying $\lambda_i\ne\lambda_j$. This is equivalent to
\begin{equation}
[\rho,A]=0,
\end{equation}
and hence (Q2) holds.

To prove (Q3), let $P=\ketbra{\psi}{\psi}$ be a pure state and choose an eigenbasis of $P$ such that $\ket{1}=\ket{\psi}$. Then the eigenvalues of $P$ are $
\lambda_1=1,\ \lambda_j=0\ (j=2,\ldots,n)$. 
In the eigenvalue representation \eqref{eq:HPskew-spectral}, all terms with $i,j\ge2$ vanish because $\lambda_i-\lambda_j=0$. Hence only the pairs $(1,j)$ with $j\ge2$ contribute. Using the boundary relation \eqref{eq:Mf-boundary},
\begin{equation}
M_f(1,0)=f(0),
\end{equation}
we obtain
\begin{equation}
\mathcal I_P^{f,\theta}(A)=f(0)^\theta\sum_{j=2}^n\frac{|A_{1j}|^2}{M_f(1,0)^\theta}=\sum_{j=2}^n|A_{1j}|^2.
\end{equation}
Since $\{\ket{j}\}_{j=1}^n$ is an orthonormal basis and $A=A^\dagger$,
\begin{equation}
\bra{\psi}A^2\ket{\psi}=\sum_{j=1}^n\bra{\psi}A\ket{j}\bra{j}A\ket{\psi}=\sum_{j=1}^n|A_{1j}|^2.
\end{equation}
Moreover, $|A_{11}|^2=|\bra{\psi}A\ket{\psi}|^2=\bra{\psi}A\ket{\psi}^2$. 
Therefore,
\begin{equation}
\mathcal I_P^{f,\theta}(A)=\sum_{j=2}^n|A_{1j}|^2=\bra{\psi}A^2\ket{\psi}-\bra{\psi}A\ket{\psi}^2=V_P(A).
\label{eq:pure-normalization}
\end{equation}
Thus the pure-state normalization (Q3) holds.

For (Q4), we use the joint convexity theorem of Hiai and Petz \cite[Theorem~7]{HiaiPetz2012}; see also \cite[Theorem~2.1]{HiaiPetz2013}.
For completeness, we briefly recall the argument.
Set
\begin{equation}
J_{P,Q}^f:=M_f(L_P,R_Q).
\end{equation}
The joint concavity of Kubo--Ando means implies that $(P,Q)\mapsto J_{P,Q}^f$ is jointly concave.
For $0<\theta\le1$, the quadratic form
\begin{equation}
(Y,X)\longmapsto
\langle X,Y^{-\theta}(X)\rangle_{\mathrm{HS}}
\end{equation}
is jointly convex and decreasing in $Y$ with respect to the operator order.
Hence
\begin{equation}
(P,Q,X)\longmapsto
I_f^\theta(P,Q,X)
=
\langle X,(J_{P,Q}^f)^{-\theta}(X)\rangle_{\mathrm{HS}}
\end{equation}
is jointly convex.
Since, for fixed $A$, the map
\begin{equation}
\rho\longmapsto
\bigl(\rho,\rho,i[\rho,A]\bigr)
\end{equation}
is affine, it follows that $\rho\mapsto\mathcal I_\rho^{f,\theta}(A)$ is convex for faithful states.
The extension to non-faithful states follows by the regularization
$\rho_\varepsilon=(1-\varepsilon)\rho+\varepsilon I/n$ and continuity.
The latter may also be seen directly from the continuous kernel representation introduced in \eqref{eq:HP-kernel}.

Finally, (Q5) follows directly from \eqref{eq:HPskew-spectral}. Under the simultaneous unitary transformation
\begin{equation}
\rho\mapsto U\rho U^\dagger,\quad A\mapsto UAU^\dagger,
\end{equation}
the eigenvalues of $\rho$ are unchanged, and the matrix elements of $UAU^\dagger$ in the transformed eigenbasis $\{U\ket{i}\}$ coincide with those of $A$ in $\{\ket{i}\}$. Hence
\begin{equation}
\mathcal I_{U\rho U^\dagger}^{f,\theta}(UAU^\dagger)=\mathcal I_\rho^{f,\theta}(A),
\label{eq:unitary-covariance}
\end{equation}
which proves (Q5).
\end{proof}

\subsection{The power-commutator family within the Hiai--Petz class}
\label{sec:power-path}

We now return to the power-commutator family \eqref{eq:Ks-intro} and show, as a first concrete application of the Hiai--Petz construction, that its intermediate regime belongs to the Hiai--Petz class. To this end, we recall the Stolarsky mean \cite{Stolarsky1975}. 
For $1/2\le s<1$, consider the one-parameter family
\begin{equation}
E_{s,1}(x,y):=\left(\frac{s(x-y)}{x^s-y^s}\right)^{1/(1-s)},\quad x,y>0,\quad x\ne y,
\label{eq:Stolarsky}
\end{equation}
with the continuous extension $E_{s,1}(x,x)=x$. This is a particular case of the two-parameter family of Stolarsky means.

The homogeneity of $E_{s,1}$ allows us to write
\begin{equation}
E_{s,1}(x,y)=y f_s(x/y),
\label{eq:Stolarsky-representation}
\end{equation}
where
\begin{equation}
f_s(t):=\left(\frac{s(t-1)}{t^s-1}\right)^{1/(1-s)},
\quad t>0,
\label{eq:fs}
\end{equation}
with the continuous extension $f_s(1)=1$. The following lemma records the properties of $f_s$ that are essential for the Hiai--Petz realization.

\begin{lemma}
\label{lem:fs-standard-regular}
For every $1/2\le s<1$, the function $f_s$ in \eqref{eq:fs} is a regular standard operator monotone function. More precisely,
\begin{equation}
f_s(1)=1,
\quad
f_s(t)=t f_s(t^{-1}),
\label{eq:fs-standard}
\end{equation}
and
\begin{equation}
f_s(0)=s^{1/(1-s)}>0.
\label{eq:fs0}
\end{equation}
Moreover, $f_{1/2}=f_{\mathrm{WY}}$.
\end{lemma}

\begin{proof}
Hiai and Petz \cite[Theorem~4.46]{HiaiPetz2014}, following the original result of Nakamura \cite{Nakamura1989}, show that
\begin{equation}
f_p(t)=\left(\frac{p(t-1)}{t^p-1}\right)^{1/(1-p)}
\end{equation}
is operator monotone for $-2\le p\le2$, with the limiting cases at $p=0,1$ understood continuously. Hence $f_s$ is operator monotone throughout the required range $1/2\le s<1$.
The normalization and symmetry relations in \eqref{eq:fs-standard} follow by continuity at $t=1$ and by direct calculation. Furthermore,
\begin{equation}
f_s(0)
=
\lim_{t\downarrow0}\left(\frac{s(t-1)}{t^s-1}\right)^{1/(1-s)}
=
s^{1/(1-s)}>0,
\end{equation}
which proves regularity. Finally,
\begin{equation}
f_{1/2}(t)
=
\left(\frac{(t-1)/2}{\sqrt{t}-1}\right)^2
=
\frac{(1+\sqrt{t})^2}{4}
=
f_{\mathrm{WY}}(t).
\end{equation}
\end{proof}

By \eqref{eq:Stolarsky-representation}, the operator mean associated with $f_s$ has scalar restriction
\begin{equation}
M_{f_s}(x,y)=y f_s(x/y)=E_{s,1}(x,y).
\label{eq:Mfs-Stolarsky}
\end{equation}
We now set
\begin{equation}
\theta_s:=2(1-s),
\label{eq:theta-s}
\end{equation}
so that $0<\theta_s\le1$ for $1/2\le s<1$. By Lemma~\ref{lem:fs-standard-regular}, $f_s$ is admissible in the Hiai--Petz framework throughout this range. The following proposition identifies the resulting Hiai--Petz skew information with the power-commutator family.
\begin{proposition}
\label{prop:power-HP}
For every $1/2\le s<1$, the power-commutator quantity \eqref{eq:Ks-intro} is a normalized Hiai--Petz skew information. More precisely,
\begin{equation}
K_s(\rho,A)=\mathcal I_\rho^{f_s,\theta_s}(A).
\label{eq:power-HP-identity}
\end{equation}
\end{proposition}

\begin{proof}
By the definition of $E_{s,1}$ and \eqref{eq:theta-s}, we have
\begin{equation}
M_{f_s}(x,y)^{\theta_s}
=
\left(\frac{s(x-y)}{x^s-y^s}\right)^2
=
\frac{s^2(x-y)^2}{(x^s-y^s)^2}.
\label{eq:Mfs-theta}
\end{equation}
On the other hand, \eqref{eq:fs0} gives
\begin{equation}
f_s(0)^{\theta_s}
=
\left(s^{1/(1-s)}\right)^{2(1-s)}
=
s^2.
\label{eq:fs0-theta}
\end{equation}
It follows that
\begin{equation}
f_s(0)^{\theta_s}
\frac{(x-y)^2}{M_{f_s}(x,y)^{\theta_s}}
=
(x^s-y^s)^2,
\label{eq:power-kernel-identity}
\end{equation}
where the identity is understood by continuous extension when $x=y$ or one of $x,y$ vanishes.

Substituting \eqref{eq:power-kernel-identity} into the eigenvalue representation \eqref{eq:HPskew-spectral}, we obtain
\begin{equation}
\mathcal I_\rho^{f_s,\theta_s}(A)
=
\sum_{i<j}(\lambda_i^s-\lambda_j^s)^2|A_{ij}|^2.
\label{eq:HP-power-spectral}
\end{equation}
Here the pairs with $\lambda_i=\lambda_j=0$ may now be included, since their contributions vanish.
Since
\begin{equation}
[\rho^s,A]_{ij}
=
(\lambda_i^s-\lambda_j^s)A_{ij},
\end{equation}
and $A$ is self-adjoint, we also have
\begin{equation}
\frac12\|[\rho^s,A]\|_{\mathrm{HS}}^2
=
\sum_{i<j}(\lambda_i^s-\lambda_j^s)^2|A_{ij}|^2.
\end{equation}
This proves \eqref{eq:power-HP-identity}.
\end{proof}
Combining Proposition~\ref{prop:power-HP} with Proposition~\ref{prop:HP-Q1-Q5}, we conclude that $K_s(\rho,A)$ satisfies (Q1)--(Q5) for every $1/2\le s<1$. Together with the elementary case $s=1$, this establishes the same conclusion for the entire range $1/2\le s\le1$.

\begin{remark}
The power-commutator family \eqref{eq:Ks-intro} is related to a broader class of skew-information-type trace functionals introduced by Fujii \cite{Fujii2005}:
\begin{equation}
S_{f,g}(\rho,A)
=
\Tr f(\rho)A g(\rho)A
-
\Tr f(\rho)g(\rho)A^2.
\end{equation}
For $f=g=x^s$, this gives
\begin{equation}
-S_{x^s,x^s}(\rho,A)
=
-\frac12\Tr[\rho^s,A]^2
=
K_s(\rho,A).
\end{equation}
Thus, from an algebraic point of view, the quantities $K_s$ are already contained in Fujii's broader two-function framework. In general, however, this framework does not impose the full set of positivity and convexity requirements considered here for intrinsic quantum uncertainty, and generic choices of $f$ and $g$ need not satisfy (Q1)--(Q5). The significance of Proposition~\ref{prop:power-HP} is that the particular family $K_s$, for $1/2\le s<1$, admits a realization within the Hiai--Petz framework, thereby establishing in particular the nontrivial convexity property (Q4), together with the other requirements (Q1)--(Q5).
\end{remark}

The realization in Proposition~\ref{prop:power-HP} is also structurally suggestive. At $s=1/2$, one has $\theta_s=1$ and $f_s=f_{\mathrm{WY}}$, whereas $\theta_s\downarrow0$ as $s\uparrow1$, with $K_s$ converging to $K_1$. Thus the power family already exhibits a path inside the Hiai--Petz framework that approaches $K_1$ at the continuous boundary $\theta=0$. This observation motivates the more general interpolation viewpoint developed in the next subsection.

\subsection{Interpolation structure: a common endpoint at $K_1$}
\label{sec:interpolation-structure}

The Hiai--Petz realization of the power family in Proposition~\ref{prop:power-HP} suggests a broader structural interpretation of the parameter $\theta$. We now show that, for each fixed regular standard operator monotone function $f$, the Hiai--Petz family provides, at the level of pairwise spectral coefficients, a geometric interpolation between Hansen's metric-adjusted skew information $I^f$ and the common endpoint $K_1$.

We now compare $K_1$, Hansen's metric-adjusted skew information $I^f$, and the Hiai--Petz skew information, whose spectral representations are given in Eqs.~\eqref{eq:K1-intro}, \eqref{eq:MASI-spectral}, and \eqref{eq:HPskew-spectral}, respectively.
We shall call the functions determining the coefficients of $|A_{ij}|^2$ in these eigenvalue representations their \emph{pairwise kernels}. Thus, for $x,y\ge0$, we set
\begin{equation}
k_{K_1}(x,y):=(x-y)^2,
\quad
k_f(x,y):=f(0)\frac{(x-y)^2}{M_f(x,y)},
\quad
k_{f,\theta}(x,y):=f(0)^\theta\frac{(x-y)^2}{M_f(x,y)^\theta},
\label{eq:HP-kernel}
\end{equation}
where the values at $(0,0)$ are understood by continuity.

\begin{proposition}[Kernel interpolation]
\label{prop:kernel-interpolation}
For every regular standard operator monotone function $f$ and every $0<\theta\le1$,
\begin{equation}
k_{f,\theta}(x,y)
=
k_{K_1}(x,y)^{1-\theta}k_f(x,y)^\theta
\label{eq:HP-geometric-interpolation}
\end{equation}
for $x\ne y$, with continuous extension to the diagonal. Consequently,
\begin{equation}
\mathcal I_\rho^{f,1}(A)=I_\rho^f(A),
\quad
\lim_{\theta\downarrow0}\mathcal I_\rho^{f,\theta}(A)=K_1(\rho,A).
\label{eq:HP-endpoints}
\end{equation}
Thus, for fixed $f$, the Hiai--Petz parameter gives a continuous pairwise geometric interpolation between $K_1$ and the metric-adjusted skew information determined by $f$.
\end{proposition}

\begin{proof}
The claim follows immediately from the definitions in \eqref{eq:HP-kernel}. 
The case $\theta=0$ is understood by continuous extension and lies outside the parameter range of the Hiai--Petz convexity theorem.
\end{proof}

\subsubsection{Position of the power-commutator family}
Proposition~\ref{prop:kernel-interpolation} reveals a simple global structure of the Hiai--Petz family. At $\theta=1$, one recovers the entire Hansen class, whereas, for each fixed $f$, decreasing $\theta$ continuously deforms the corresponding metric-adjusted skew information toward $K_1$. All such fixed-$f$ paths meet at the common endpoint $K_1$ as $\theta\downarrow0$. Thus $K_1$ serves as a universal endpoint of the Hiai--Petz interpolation. This interpolation is geometric at the level of pairwise spectral kernels; in general, however, the full functionals themselves are not related by a geometric mean.

The power-commutator family $K_s$ obtained above has a particularly simple position within this global structure. By Proposition~\ref{prop:power-HP},
\begin{equation}
K_s(\rho,A)
=
\mathcal I_\rho^{f_s,\theta_s}(A),
\quad
\theta_s=2(1-s),
\quad
\frac12\le s<1,
\label{eq:Ks-HP-position}
\end{equation}
where $f_s$ is the Stolarsky representing function introduced above. Combining this with Proposition~\ref{prop:kernel-interpolation} gives
\begin{equation}
k_{K_s}(x,y)
=
k_{K_1}(x,y)^{2s-1}
k_{f_s}(x,y)^{2(1-s)}.
\label{eq:Ks-kernel-interpolation}
\end{equation}
Hence, for each fixed $1/2<s<1$, $K_s$ lies on the Hiai--Petz interpolation path connecting the metric-adjusted skew information $I^{f_s}$ to the common endpoint $K_1$.

As $s$ varies, not only the interpolation parameter $\theta_s=2(1-s)$ but also the representing function $f_s$ changes. Therefore, the power family does not follow a single fixed-$f$ Hiai--Petz interpolation path. Rather, it traces a distinguished curve through the Hiai--Petz family: at $s=1/2$ one has $f_{1/2}=f_{\mathrm{WY}}$ and
\begin{equation}
K_{1/2}(\rho,A)=I_\rho^{\mathrm{WY}}(A),
\end{equation}
whereas
\begin{equation}
\lim_{s\uparrow1}K_s(\rho,A)=K_1(\rho,A).
\end{equation}
Thus the power-commutator family provides a continuous path from the Wigner--Yanase skew information to $K_1$, while simultaneously selecting, for each intermediate $s$, a distinguished point on a different fixed-$f_s$ Hiai--Petz interpolation path.

\subsubsection{Interpolation to the SLD quantum Fisher information}

A particularly important fixed-$f$ path is obtained from the SLD function
\begin{equation}
f_{\mathrm{SLD}}(t)=\frac{1+t}{2}.
\end{equation}
In contrast to the power family, the SLD interpolation is a genuine fixed-$f$ path: the representing function $f_{\mathrm{SLD}}$ is kept fixed while only $\theta$ varies.
As shown in Sec.~\ref{sec:single-observable-uncertainty}, one quarter of the SLD quantum Fisher information is the maximal intrinsic quantum uncertainty among functionals satisfying (Q3) and (Q4). The corresponding Hiai--Petz path therefore connects the common endpoint $K_1$ directly to this maximal element of the Hansen class.

Since
\begin{equation}
M_{f_{\mathrm{SLD}}}(x,y)=\frac{x+y}{2},
\quad
f_{\mathrm{SLD}}(0)=\frac12,
\end{equation}
the interpolation takes the particularly simple form
\begin{equation}
\mathcal I_\rho^{\mathrm{SLD},\theta}(A)
:=
\mathcal I_\rho^{f_{\mathrm{SLD}},\theta}(A)
=
\sum_{\substack{i<j\\ \lambda_i+\lambda_j>0}}
\frac{(\lambda_i-\lambda_j)^2}
{(\lambda_i+\lambda_j)^\theta}
|A_{ij}|^2.
\label{eq:HP-SLD-family}
\end{equation}
Its endpoints are
\begin{equation}
\lim_{\theta\downarrow0}
\mathcal I_\rho^{\mathrm{SLD},\theta}(A)
=
K_1(\rho,A),
\quad
\mathcal I_\rho^{\mathrm{SLD},1}(A)
=
\frac14F_Q^{\mathrm{SLD}}(\rho,A).
\label{eq:HP-SLD-endpoints}
\end{equation}
Thus the Hiai--Petz family provides a continuous interpolation between the Hilbert--Schmidt commutator $K_1$ and the strongest single-observable quantum uncertainty allowed by the present axioms.

\subsection{Deformed additivity}
\label{sec:deformed-additivity}

The definition of intrinsic quantum uncertainty used here does not require a composition axiom. Nevertheless, the Hiai--Petz parameter produces a controlled deformation of the familiar composition properties of metric-adjusted skew information.

As mentioned above, the original Wigner--Yanase skew information and Hansen's metric-adjusted skew informations are additive for statistically independent systems. In general, however, this additivity is no longer retained by the normalized Hiai--Petz skew information. A simple illustration is provided by the power-commutator family \eqref{eq:Ks-intro}. For $1/2<s\le1$, it is generally nonadditive. Indeed, for product states $\rho_1\otimes\rho_2$ and additive observables $A_1\otimes I+I\otimes A_2$, we have
\begin{equation}
[(\rho_1\otimes\rho_2)^s,A_1\otimes I+I\otimes A_2]
=
[\rho_1^s,A_1]\otimes\rho_2^s
+
\rho_1^s\otimes[\rho_2^s,A_2].
\end{equation}
The cross term in the squared Hilbert--Schmidt norm vanishes, since
\begin{equation}
\Tr\bigl([\rho_1^s,A_1]^\dagger\rho_1^s\bigr)=0,
\quad
\Tr\bigl((\rho_2^s)^\dagger[\rho_2^s,A_2]\bigr)=0.
\end{equation}
Consequently,
\begin{equation}
K_s(\rho_1\otimes\rho_2,A_1\otimes I+I\otimes A_2)
=
\Tr(\rho_2^{2s})K_s(\rho_1,A_1)
+
\Tr(\rho_1^{2s})K_s(\rho_2,A_2).
\label{eq:Ks-deformed-additivity}
\end{equation}
For $s=1/2$, the factors $\Tr(\rho_1^{2s})$ and $\Tr(\rho_2^{2s})$ are both equal to one, and \eqref{eq:Ks-deformed-additivity} reduces to the ordinary additivity of the Wigner--Yanase skew information. For $1/2<s\le1$, however, these factors are generally different from one for mixed states, and ordinary additivity fails.

Nevertheless, relation \eqref{eq:Ks-deformed-additivity} suggests that additivity is not simply lost, but is instead replaced by a natural deformed composition law. In fact, the entire family of normalized Hiai--Petz skew informations satisfies the following power-trace-weighted additivity relation.

\begin{proposition}
\label{prop:deformed-additivity}
Let $f$ be a regular standard operator monotone function and let $0<\theta\le1$. For density operators $\rho_1$ and $\rho_2$ and observables $A_1$ and $A_2$,
\begin{equation}
\mathcal I_{\rho_1\otimes\rho_2}^{f,\theta}
\left(A_1\otimes I+I\otimes A_2\right)
=
\Tr\left(\rho_2^{2-\theta}\right)
\mathcal I_{\rho_1}^{f,\theta}(A_1)
+
\Tr\left(\rho_1^{2-\theta}\right)
\mathcal I_{\rho_2}^{f,\theta}(A_2).
\label{eq:deformed-additivity}
\end{equation}
Equivalently, using the Tsallis entropy \eqref{eq:Tsallis},
\begin{align}
&\mathcal I_{\rho_1\otimes\rho_2}^{f,\theta}
\left(A_1\otimes I+I\otimes A_2\right)
\nonumber\\
&\quad=
\mathcal I_{\rho_1}^{f,\theta}(A_1)
+
\mathcal I_{\rho_2}^{f,\theta}(A_2)
-(1-\theta)
\left[
S_{2-\theta}(\rho_2)\mathcal I_{\rho_1}^{f,\theta}(A_1)
+
S_{2-\theta}(\rho_1)\mathcal I_{\rho_2}^{f,\theta}(A_2)
\right].
\label{eq:Tsallis-additivity-deficit}
\end{align}
In particular, at $\theta=1$, corresponding to Hansen's metric-adjusted skew information, this reduces to the ordinary additivity law \eqref{eq:Q6-additivity}.
\end{proposition}

\begin{proof}
Let
\begin{equation}
\rho_1=\sum_i\lambda_i\ketbra{i}{i},
\quad
\rho_2=\sum_a\mu_a\ketbra{a}{a}.
\end{equation}
Then $\rho_1\otimes\rho_2$ has eigenvalues $\lambda_i\mu_a$. The matrix elements of $A_1\otimes I+I\otimes A_2$ between the product eigenvectors $\ket{i,a}$ and $\ket{j,b}$ are
\begin{equation}
(A_1)_{ij}\delta_{ab}
+
\delta_{ij}(A_2)_{ab}.
\end{equation}
Hence the two contributions are supported on disjoint sets of index pairs unless both $i=j$ and $a=b$, in which case the commutator contribution vanishes. Therefore there are no cross terms.

For the $A_1\otimes I$ contribution, using the homogeneity
\begin{equation}
M_f(tx,ty)=tM_f(x,y),
\end{equation}
we have
\begin{equation}
\frac{(\lambda_i\mu_a-\lambda_j\mu_a)^2}
{M_f(\lambda_i\mu_a,\lambda_j\mu_a)^\theta}
=
\mu_a^{2-\theta}
\frac{(\lambda_i-\lambda_j)^2}
{M_f(\lambda_i,\lambda_j)^\theta}.
\end{equation}
Summing over $a$ yields the factor $\Tr(\rho_2^{2-\theta})$. The $I\otimes A_2$ contribution is analogous and gives the factor $\Tr(\rho_1^{2-\theta})$. 
Substitution into the spectral formula \eqref{eq:HPskew-spectral} proves \eqref{eq:deformed-additivity}.
\end{proof}

\begin{remark}
\label{rem:beyond-MASI}
The deformed composition law in Proposition~\ref{prop:deformed-additivity} reflects the homogeneity of the pairwise spectral kernel. 
Since the operator mean $M_f$ is homogeneous of degree one, \eqref{eq:HP-kernel} implies
\begin{equation}
k_{f,\theta}(tx,ty)
=
t^{2-\theta}k_{f,\theta}(x,y).
\label{eq:HP-kernel-homogeneity}
\end{equation}
Thus the metric-adjusted boundary $\theta=1$ is precisely the case of degree-one homogeneity.

Indeed, if $\rho_2=\sum_a\mu_a\ketbra{a}{a}$, the $A_1\otimes I$ contribution involves the eigenvalue pairs $(\mu_a\lambda_i,\mu_a\lambda_j)$ and therefore acquires the factor $\mu_a^{2-\theta}$. 
Summing over $a$ gives $\Tr(\rho_2^{2-\theta})$, which explains the power-trace factor in \eqref{eq:deformed-additivity}. 
In this sense, the deformed additivity is the tensor-product manifestation of the homogeneity degree $2-\theta$.

In particular, setting $A_2=0$ gives
\begin{equation}
\mathcal I_{\rho_1\otimes\rho_2}^{f,\theta}(A_1\otimes I)
=
\Tr(\rho_2^{2-\theta})
\mathcal I_{\rho_1}^{f,\theta}(A_1).
\label{eq:HP-ancilla-scaling}
\end{equation}
For $0<\theta<1$ and mixed $\rho_2$, the prefactor
$\Tr(\rho_2^{2-\theta})$ is strictly smaller than one, whereas every metric-adjusted skew information is invariant under adjoining an uncorrelated ancillary state. 
Hence an interior Hiai--Petz functional cannot coincide with a Hansen metric-adjusted skew information as a dimension-independent family.
\end{remark}

\section{Optimal single-observable uncertainty relations}
\label{sec:optimal-single-observable}

\subsection{Fixed-state optimization and the classical contribution}
The universal relation \eqref{eq:single-observable-UR},
\begin{equation}
V_\rho(A)\ge Q_\rho(A),
\end{equation}
follows solely from the pure-state normalization (Q3) and convexity (Q4). 
For a fixed mixed state $\rho$, however, the coefficient multiplying $Q_\rho(A)$ need not be optimal. This motivates us to consider the strongest state-dependent refinement of the form
\begin{equation}
V_\rho(A)\ge c(\rho)Q_\rho(A),
\label{eq:state-dependent-UR}
\end{equation}
where $c(\rho)$ depends only on the state. For a fixed state $\rho$, we define the optimal coefficient by
\begin{equation}
c_Q^{\mathrm{opt}}(\rho)
:=
\sup\left\{
c\ge0:
V_\rho(A)\ge cQ_\rho(A)
\text{ for all }A\in\mathcal L_{\rm sa}(\HA)
\right\}.
\label{eq:copt-general}
\end{equation}
Since \eqref{eq:single-observable-UR} always holds, one necessarily has $c_Q^{\mathrm{opt}}(\rho)\ge1$.

In our recent work \cite{YamashitaMayumiKimura2026}, this fixed-state optimization was carried out for several representative intrinsic quantum uncertainties, including the Wigner--Yanase and Wigner--Yanase--Dyson skew informations, the SLD quantum Fisher information, and the power-commutator family $K_s$. 
We further identified a classical contribution to the uncertainty, denoted by $V_\rho^{\mathrm{cl}}(A)$, and established the stronger uncertainty relations
\begin{equation}
V_\rho(A)
\ge
V_\rho^{\mathrm{cl}}(A)
+
c_Q^{\mathrm{opt}}(\rho)Q_\rho(A).
\label{eq:optimal-UR-with-classical}
\end{equation}
Remarkably, the same optimal coefficient $c_Q^{\mathrm{opt}}(\rho)$ appears even after the additional classical contribution is included. 
The purpose of the present section is to unify these optimal single-observable uncertainty relations within the framework of Hiai--Petz skew information.

We first recall the classical contribution to the variance. Let $\rho=\sum_{i=1}^n\lambda_i\ketbra{i}{i}$ be an eigenvalue decomposition of $\rho$, and denote the corresponding orthonormal eigenbasis by $\mathcal B=\{\ket{i}\}_{i=1}^n$. 
For a fixed eigenbasis $\mathcal B$, we define
\begin{equation}
V_{\rho,\mathcal B}^{\mathrm{cl}}(A)
:=
\sum_{i=1}^n
\lambda_i
\left(
\bra{i}A\ket{i}
-
\langle A\rangle_\rho
\right)^2.
\label{eq:classical-variance-basis}
\end{equation}
When $\rho$ has degenerate eigenvalues, this quantity may depend on the choice of eigenbasis within the degenerate eigenspaces. 
We therefore define the classical uncertainty by maximizing over all eigenbases of $\rho$:
\begin{equation}
V_\rho^{\mathrm{cl}}(A)
:=
\max_{\mathcal B\in\mathfrak B(\rho)}
V_{\rho,\mathcal B}^{\mathrm{cl}}(A),
\label{eq:classical-variance}
\end{equation}
where $\mathfrak B(\rho)$ denotes the set of all orthonormal eigenbases of $\rho$. 
The maximum is attained by choosing, within each eigenspace of $\rho$, an orthonormal basis that diagonalizes the compression of $A$ to that eigenspace. Indeed, this maximization admits the basis-independent characterization
\begin{equation}
V_\rho^{\mathrm{cl}}(A)
=
V_\rho\bigl(\mathcal P_\rho(A)\bigr),
\label{eq:classical-variance-pinching}
\end{equation}
where, if $\rho=\sum_{\alpha=1}^m r_\alpha P_\alpha$ is the spectral decomposition into distinct eigenvalues and the corresponding spectral projections,
\begin{equation}
\mathcal P_\rho(A):=\sum_{\alpha=1}^mP_\alpha A P_\alpha.
\end{equation}

The pinching $\mathcal P_\rho(A)$ removes precisely the components of $A$ connecting distinct eigenspaces of $\rho$, while retaining the part that commutes with $\rho$. Thus, $V_\rho^{\mathrm{cl}}(A)$ quantifies the maximal contribution to the variance that can be attributed to the commuting, or classical, part of the observable relative to the state.
For further details and the proof of the equivalence \eqref{eq:classical-variance-pinching}, see Ref.~\cite{YamashitaMayumiKimura2026}.

\subsection{Sharp Hiai--Petz uncertainty relation}

We can now state the main result of this section.
\begin{theorem}[Optimal Hiai--Petz single-observable uncertainty relation]
\label{thm:optimal-HP}
Let $f$ be a regular standard operator monotone function, and let $0<\theta\le1$. Then, for any density operator $\rho$, every observable $A$ satisfies the sharp uncertainty relation
\begin{equation}
V_\rho(A)
\ge
V_\rho^{\mathrm{cl}}(A)
+
c_{f,\theta}^{\mathrm{opt}}(\rho)
\mathcal I_\rho^{f,\theta}(A),
\label{eq:optimal-HP-UR}
\end{equation}
where, for a nonmaximally mixed state, the optimal coefficient is given by
\begin{equation}
c_{f,\theta}^{\mathrm{opt}}(\rho)
=
\frac{
(\lambda_{\max}+\lambda_{\min})
M_f(\lambda_{\max},\lambda_{\min})^\theta
}{
f(0)^\theta
(\lambda_{\max}-\lambda_{\min})^2
}.
\label{eq:copt-HP}
\end{equation}
Here $\lambda_{\max}$ and $\lambda_{\min}$ denote the largest and smallest eigenvalues of $\rho$, respectively. For the maximally mixed state, the second term on the right-hand side of \eqref{eq:optimal-HP-UR} is understood to vanish. In this case, $V_\rho^{\mathrm{cl}}(A)=V_\rho(A)$, and hence \eqref{eq:optimal-HP-UR} is saturated for every observable.
\end{theorem}

\begin{remark}
Remarkably, the optimal coefficient in Theorem~\ref{thm:optimal-HP} is determined solely by the smallest and largest eigenvalues of $\rho$, although the Hiai--Petz skew information itself generally depends on the full spectrum.
Both $V_\rho(A)-V_\rho^{\mathrm{cl}}(A)$ and $\mathcal I_\rho^{f,\theta}(A)$ decompose into sums over the same eigenvalue pairs, with nonnegative coefficients multiplying $|A_{ij}|^2$.
Hence the optimization reduces to minimizing the corresponding pairwise ratio.
The nontrivial point is that this minimum is always attained by the extremal pair $(\lambda_{\max},\lambda_{\min})$.
A similar dependence on extremal eigenvalues was observed in the optimization of Robertson-type uncertainty relations \cite{KimuraMayumiYamashita}.
\end{remark}

\begin{proof}
We first consider the maximally mixed state. In this case, $\rho$ commutes with every observable, and hence $\mathcal I_\rho^{f,\theta}(A)=0$. Moreover, since every orthonormal basis is an eigenbasis of $\rho$, the basis in the definition of $V_\rho^{\mathrm{cl}}(A)$ may be chosen to diagonalize $A$. It follows that $V_\rho^{\mathrm{cl}}(A)=V_\rho(A)$, and therefore \eqref{eq:optimal-HP-UR} is saturated.

Suppose now that $\rho$ is nonmaximally mixed. Choose an eigenvalue decomposition $\rho=\sum_i\lambda_i\ketbra{i}{i}$ whose eigenbasis attains the maximum in \eqref{eq:classical-variance}. By construction, within each degenerate eigenspace the basis diagonalizes the corresponding compression of $A$. Hence $A_{ij}=0$ whenever $i\ne j$ and $\lambda_i=\lambda_j$. A direct computation then gives
\begin{equation}
V_\rho(A)
=
V_\rho^{\mathrm{cl}}(A)
+
\sum_{\substack{i<j\\ \lambda_i\ne\lambda_j}}
(\lambda_i+\lambda_j)|A_{ij}|^2.
\label{eq:variance-classical-decomposition}
\end{equation}
On the other hand, the eigenvalue representation \eqref{eq:HPskew-spectral} gives
\begin{equation}
\mathcal I_\rho^{f,\theta}(A)
=
f(0)^\theta
\sum_{\substack{i<j\\ \lambda_i\ne\lambda_j}}
\frac{(\lambda_i-\lambda_j)^2}
{M_f(\lambda_i,\lambda_j)^\theta}
|A_{ij}|^2.
\label{eq:HP-distinct-eigenvalues}
\end{equation}

We next show that the optimal pairwise coefficient is always attained by the extremal eigenvalues. For $x>y\ge0$, define
\begin{equation}
C_{f,\theta}(x,y)
:=
\frac{(x+y)M_f(x,y)^\theta}
{f(0)^\theta(x-y)^2}.
\label{eq:pairwise-coefficient}
\end{equation}
Writing $r=y/x$, the homogeneity and symmetry of $M_f$ give $M_f(x,y)=xf(r)$, and hence
\begin{equation}
C_{f,\theta}(x,y)
=
\frac{x^{\theta-1}}{f(0)^\theta}
\frac{(1+r)f(r)^\theta}{(1-r)^2}.
\label{eq:pairwise-ratio}
\end{equation}
Since $f$ is operator monotone, it is increasing as a scalar function. Therefore
\begin{equation}
r\longmapsto
\frac{(1+r)f(r)^\theta}{(1-r)^2}
\end{equation}
is increasing on $[0,1)$. For any pair of distinct eigenvalues $x>y$, we have
\begin{equation}
x\le\lambda_{\max},
\quad
\frac{y}{x}\ge\frac{\lambda_{\min}}{\lambda_{\max}}.
\end{equation}
Since $0<\theta\le1$, we also have $x^{\theta-1}\ge\lambda_{\max}^{\theta-1}$. It follows from \eqref{eq:pairwise-ratio} that
\begin{equation}
C_{f,\theta}(x,y)
\ge
C_{f,\theta}(\lambda_{\max},\lambda_{\min})
=
c_{f,\theta}^{\mathrm{opt}}(\rho).
\label{eq:extremal-pair}
\end{equation}
Thus, for every pair with $\lambda_i\ne\lambda_j$,
\begin{equation}
\lambda_i+\lambda_j
\ge
c_{f,\theta}^{\mathrm{opt}}(\rho)
f(0)^\theta
\frac{(\lambda_i-\lambda_j)^2}
{M_f(\lambda_i,\lambda_j)^\theta}.
\end{equation}
Multiplying by $|A_{ij}|^2$, summing over all such pairs, and using \eqref{eq:variance-classical-decomposition} and \eqref{eq:HP-distinct-eigenvalues} yields \eqref{eq:optimal-HP-UR}.

It remains to prove optimality. Let $\ket{\psi_{\max}}$ and $\ket{\psi_{\min}}$ be unit eigenvectors corresponding to $\lambda_{\max}$ and $\lambda_{\min}$, respectively, and define
\begin{equation}
A_*=
\ketbra{\psi_{\max}}{\psi_{\min}}
+
\ketbra{\psi_{\min}}{\psi_{\max}}.
\label{eq:optimal-observable}
\end{equation}
Since $A_*$ connects two distinct eigenspaces of $\rho$, one has $\mathcal P_\rho(A_*)=0$ and therefore $V_\rho^{\mathrm{cl}}(A_*)=0$. Moreover, only the pair $(\lambda_{\max},\lambda_{\min})$ contributes to both \eqref{eq:variance-classical-decomposition} and \eqref{eq:HP-distinct-eigenvalues}. Equality is therefore attained in \eqref{eq:optimal-HP-UR}. Hence the coefficient \eqref{eq:copt-HP} cannot be increased.
\end{proof}

Dropping the nonnegative classical contribution from \eqref{eq:optimal-HP-UR}, we immediately obtain the following sharp bound.

\begin{corollary}\label{cor:optimal-HP-simple}
Let $f$ be a regular standard operator monotone function, let $0<\theta\le1$, and let $\rho$ be a nonmaximally mixed state.
Then every observable $A$ satisfies
\begin{equation}
V_\rho(A)
\ge
c_{f,\theta}^{\mathrm{opt}}(\rho)
\mathcal I_\rho^{f,\theta}(A),
\label{eq:optimal-HP-UR-simple}
\end{equation}
where $c_{f,\theta}^{\mathrm{opt}}(\rho)$ is given by \eqref{eq:copt-HP}.
The coefficient is optimal for the fixed state $\rho$.
\end{corollary}

Indeed, the observable $A_*$ in \eqref{eq:optimal-observable} satisfies $V_\rho^{\mathrm{cl}}(A_*)=0$ and saturates \eqref{eq:optimal-HP-UR-simple}.
Thus, the classical contribution $V_\rho^{\mathrm{cl}}(A)$ can be retained without reducing the optimal coefficient of the intrinsic quantum uncertainty term.
For the maximally mixed state, $\mathcal I_\rho^{f,\theta}(A)=0$ for every observable, so \eqref{eq:optimal-HP-UR-simple} reduces to the trivial inequality $V_\rho(A)\ge0$ and no finite optimal coefficient is singled out.

\begin{remark}
An interesting distinction from Hansen's metric-adjusted class appears for non-faithful states.
If $\lambda_{\min}=0$, the optimal coefficient reduces to
\begin{equation}
c_{f,\theta}^{\mathrm{opt}}(\rho)
=
\lambda_{\max}^{\theta-1}.
\end{equation}
On the metric-adjusted boundary $\theta=1$, one therefore has $c_{f,1}^{\mathrm{opt}}(\rho)=1$, so that no state-dependent improvement of the coefficient occurs, even for a rank-deficient mixed state.
By contrast, for $0<\theta<1$, every rank-deficient mixed state satisfies $\lambda_{\max}<1$ and hence $c_{f,\theta}^{\mathrm{opt}}(\rho)>1$.
This improvement persists at the continuous endpoint $\theta=0$, corresponding to $K_1$, where $c^{\mathrm{opt}}(\rho)=1/\lambda_{\max}$.
Thus the non-metric Hiai--Petz regime, together with its $K_1$ endpoint, exhibits a genuine state-dependent enhancement even in the presence of a zero eigenvalue.
For a pure state, $\lambda_{\max}=1$, and the coefficient remains equal to one throughout.
\end{remark}

Theorem~\ref{thm:optimal-HP} recovers all the optimal relations obtained in Ref.~\cite{YamashitaMayumiKimura2026} for the Wigner--Yanase and Wigner--Yanase--Dyson skew informations, the SLD quantum Fisher information, and the power-commutator family.

For $0<\alpha<1$, the Wigner--Yanase--Dyson skew information corresponds to the operator monotone function \eqref{eq:f-WYD}, whose associated mean is
\begin{equation}
M_{f_\alpha^{\mathrm{WYD}}}(x,y)
=
\alpha(1-\alpha)
\frac{(x-y)^2}
{(x^\alpha-y^\alpha)(x^{1-\alpha}-y^{1-\alpha})}.
\end{equation}
Setting $\theta=1$ in \eqref{eq:copt-HP} gives
\begin{equation}
c_{\mathrm{WYD},\alpha}^{\mathrm{opt}}(\rho)
=
\frac{\lambda_{\max}+\lambda_{\min}}
{(\lambda_{\max}^\alpha-\lambda_{\min}^\alpha)
(\lambda_{\max}^{1-\alpha}-\lambda_{\min}^{1-\alpha})}.
\label{eq:copt-WYD}
\end{equation}
Hence
\begin{equation}
V_\rho(A)
\ge
V_\rho^{\mathrm{cl}}(A)
+
\frac{\lambda_{\max}+\lambda_{\min}}
{(\lambda_{\max}^\alpha-\lambda_{\min}^\alpha)
(\lambda_{\max}^{1-\alpha}-\lambda_{\min}^{1-\alpha})}
I_{\rho,\alpha}^{\mathrm{WYD}}(A).
\label{eq:optimal-WYD}
\end{equation}
At $\alpha=1/2$, this reduces to the Wigner--Yanase relation
\begin{equation}
V_\rho(A)
\ge
V_\rho^{\mathrm{cl}}(A)
+
\frac{\lambda_{\max}+\lambda_{\min}}
{(\sqrt{\lambda_{\max}}-\sqrt{\lambda_{\min}})^2}
I_\rho^{\mathrm{WY}}(A).
\label{eq:optimal-WY}
\end{equation}

For the SLD function \eqref{eq:f-SLD},
\begin{equation}
M_{f_{\mathrm{SLD}}}(x,y)=\frac{x+y}{2}.
\end{equation}
Thus, at $\theta=1$, Eq.~\eqref{eq:copt-HP} gives
\begin{equation}
c_{\mathrm{SLD}}^{\mathrm{opt}}(\rho)
=
\frac{(\lambda_{\max}+\lambda_{\min})^2}
{(\lambda_{\max}-\lambda_{\min})^2}.
\label{eq:copt-SLD}
\end{equation}
Since $I_\rho^{f_{\mathrm{SLD}}}(A)=\frac14F_Q^{\mathrm{SLD}}(\rho,A)$, the corresponding sharp relation is
\begin{equation}
V_\rho(A)
\ge
V_\rho^{\mathrm{cl}}(A)
+
\frac{1}{4}\frac{(\lambda_{\max}+\lambda_{\min})^2}
{(\lambda_{\max}-\lambda_{\min})^2}
F_Q^{\mathrm{SLD}}(\rho,A).
\label{eq:optimal-SLD}
\end{equation}

For $1/2\le s<1$, Proposition~\ref{prop:power-HP} identifies $K_s$ with the Hiai--Petz skew information determined by $f_s$ in \eqref{eq:fs} and $\theta_s=2(1-s)$. Using \eqref{eq:Mfs-theta} and \eqref{eq:fs0-theta} in \eqref{eq:copt-HP}, we obtain
\begin{equation}
c_{K_s}^{\mathrm{opt}}(\rho)
=
\frac{\lambda_{\max}+\lambda_{\min}}
{(\lambda_{\max}^s-\lambda_{\min}^s)^2}.
\label{eq:copt-Ks}
\end{equation}
Consequently,
\begin{equation}
V_\rho(A)
\ge
V_\rho^{\mathrm{cl}}(A)
+
\frac{\lambda_{\max}+\lambda_{\min}}
{(\lambda_{\max}^s-\lambda_{\min}^s)^2}
K_s(\rho,A),
\quad
\frac12\le s\le1.
\label{eq:optimal-Ks}
\end{equation}
At $s=1/2$, this reproduces \eqref{eq:optimal-WY}, while at $s=1$ it gives
\begin{equation}
V_\rho(A)
\ge
V_\rho^{\mathrm{cl}}(A)
+
\frac{\lambda_{\max}+\lambda_{\min}}
{(\lambda_{\max}-\lambda_{\min})^2}
K_1(\rho,A).
\label{eq:optimal-K1}
\end{equation}
Thus the power family yields a continuous family of sharp uncertainty relations connecting the Wigner--Yanase bound to the $K_1$ bound.

The second interpolation structure discussed above also carries over directly to the sharp uncertainty relations. 
For the fixed-SLD Hiai--Petz path \eqref{eq:HP-SLD-family}, Theorem~\ref{thm:optimal-HP} gives
\begin{equation}
V_\rho(A)
\ge
V_\rho^{\mathrm{cl}}(A)
+
\frac{(\lambda_{\max}+\lambda_{\min})^{1+\theta}}
{(\lambda_{\max}-\lambda_{\min})^2}
\mathcal I_\rho^{\mathrm{SLD},\theta}(A),
\quad
0<\theta\le1.
\label{eq:optimal-HP-SLD}
\end{equation}
The coefficient is optimal for the fixed state $\rho$. At $\theta=1$, this reduces to \eqref{eq:optimal-SLD}, whereas the continuous limit $\theta\downarrow0$ yields \eqref{eq:optimal-K1}. Thus the fixed-SLD Hiai--Petz path provides a continuous family of sharp uncertainty relations connecting the $K_1$ bound to the optimal SLD quantum Fisher information bound.

\section{Conclusion and Discussion}
\label{sec:conclusion}

We have developed a class of intrinsic quantum uncertainty measures from normalized commutator restrictions of the Hiai--Petz quasi-entropy-type functionals. 
The construction preserves positivity, faithfulness to noncommutativity, pure-state normalization, convexity, and unitary covariance, while ordinary additivity is not imposed as a defining requirement. 
Hansen's metric-adjusted skew informations are recovered at the boundary $\theta=1$, whereas the interior $0<\theta<1$ provides genuinely nonadditive extensions with a controlled composition law.

Our main quantitative result is the exact solution of the fixed-state optimization problem for the full Hiai--Petz class. 
Even after the maximal classical contribution to the variance is retained, the optimal coefficient admits a closed form depending only on the largest and smallest eigenvalues of the state. 
This yields sharp single-observable uncertainty relations for Hansen's metric-adjusted skew informations and for the nonadditive Hiai--Petz family. In particular, it includes the SLD quantum Fisher information and the power-commutator family as important special cases.

A central structural example is the power-commutator family $K_s$. We showed that, for $1/2\le s<1$, $K_s$ can be realized as a Hiai--Petz skew information through a family of Stolarsky operator means. 
This representation proves, in particular, the nontrivial convexity of $K_s$ for $1/2<s<1$. 
It also places the power family naturally within the Hiai--Petz interpolation structure: $K_{1/2}$ coincides with the Wigner--Yanase skew information, while $K_s$ approaches $K_1$ as $s\uparrow1$.

More generally, for every fixed regular standard operator monotone function $f$, the Hiai--Petz pairwise kernel gives a geometric interpolation between the kernel of the metric-adjusted skew information $I^f$ at $\theta=1$ and that of $K_1$ at the continuous endpoint $\theta=0$. 
Thus the entire Hansen class forms one boundary of the Hiai--Petz family, while all fixed-$f$ paths meet at the common endpoint $K_1$. The power family traces a distinguished curve through this structure because both the Stolarsky function $f_s$ and the parameter $\theta_s=2(1-s)$ vary with $s$. 
By contrast, fixing the SLD representing function gives a direct interpolation between $K_1$ and one quarter of the SLD quantum Fisher information, the maximal intrinsic quantum uncertainty within the present framework. 
The exact power-trace-weighted composition law further distinguishes the interior $0<\theta<1$ from Hansen's additive class and explains why $K_s$ for $1/2<s<1$, as well as the endpoint $K_1$ when regarded as a dimension-independent family, lies outside the metric-adjusted class.

The extremal-eigenvalue structure of the optimal coefficient suggests a broader question concerning the spectral sectors selected by sharp uncertainty relations.
In the present setting, the pairwise decomposition reduces the optimization to a single pair of eigenspaces, while the nontrivial result is that the minimizing pair is always associated with $\lambda_{\max}$ and $\lambda_{\min}$.
A similar extremal-eigenvalue reduction occurs in optimal Robertson-type relations \cite{KimuraMayumiYamashita}, whereas uncertainty relations involving the absolute value of the commutator can instead select the two smallest eigenvalues \cite{KimuraMayumiOhno}.
It would therefore be interesting to understand what structural properties determine which spectral sector governs a sharp uncertainty relation, particularly in problems that do not admit the pairwise decomposition available in the present setting.
Such a perspective may reveal a broader principle underlying spectral reductions in optimal uncertainty relations.

Finally, it would also be interesting to clarify the operational meaning of the Hiai--Petz interpolation parameter and to determine whether the different interpolation structures identified here correspond to distinct information-theoretic or metrological tasks.

\section*{Acknowledgments}

The author is deeply grateful to Professor Fumio Hiai for his mathematical influence, encouragement, and many years of kindness. 
The author also thanks Haruki Yamashita for valuable discussions and comments leading to this study.
This work was supported by JSPS KAKENHI Grant No. 24K06873.

\bibliographystyle{unsrtnat}
\bibliography{refsHiai}

\end{document}